\documentclass[aps,pra,twocolumn,nofootinbib,a4paper,amsfonts,amssymb,amsmath,floatfix]{revtex4-2}
\pdfoutput=1
\usepackage[utf8]{inputenc}                 % allow utf-8 input
\usepackage[american]{babel}                % the american english
\usepackage[
  colorlinks=true,
  allcolors=teal,
  breaklinks=true
]{hyperref}                                 % hyperlinks
\usepackage[
  caption=false
  ]{subfig}                                 % \subfloat{}
\usepackage{microtype}	                    % beautiful spacing
\usepackage{tikz}	                          % code figures
\tikzstyle{vertex} = [fill,draw,circle,minimum size=3,inner sep=0pt]
\usepackage{mathrsfs}                       % for \mathscr font
\usepackage{algorithm}                      % for algorithm environment
\usepackage[noEnd=true]{algpseudocodex}     % for algorithmic environment
\algrenewcommand\algorithmicrequire{\textbf{Input: }}
\algrenewcommand\algorithmicensure{\textbf{Output: }}
\usepackage{amsthm}	                        % define theorems, lemmas, etc.
\newtheorem{theorem}{Theorem}               % theorem environment
\theoremstyle{definition}		                % style for definition environment
\newtheorem{definition}{Definition}         % definition environment

\newcommand{\cptp}{\operatorname{CPTP}}     % cptp map
\DeclareMathOperator{\Tr}{Tr}               % trace
\DeclareMathOperator{\Pa}{Pa}               % parents
\DeclareMathOperator{\Ch}{Ch}               % children
\newcommand{\arc}{\!\!\rightarrowtail\!\!}  % arc

\begin{document}
\title{Flow of dynamical causal structures with an application to correlations}
\author{Ämin Baumeler}
	\affiliation{Facoltà di scienze informatiche, Università della Svizzera italiana, 6900 Lugano, Switzerland}
	\affiliation{Facoltà indipendente di Gandria, 6978 Gandria, Switzerland}
\author{Stefan Wolf}
	\affiliation{Facoltà di scienze informatiche, Università della Svizzera italiana, 6900 Lugano, Switzerland}
	\affiliation{Facoltà indipendente di Gandria, 6978 Gandria, Switzerland}

\begin{abstract}
  \noindent
  Causal models capture cause-effect relations both qualitatively---via the graphical {\em causal structure---}and quantitatively---via the {\em model parameters.}
  They offer a~powerful framework for analyzing and constructing processes.
  Here, we introduce a~tool---{\em the flow of causal structures}---to visualize and explore the {\em dynamical aspect\/} of classical-deterministic processes, arguably like those present in general relativity.
  The flow describes all possible ways in which the causal structure of a~process can evolve.
  We also present an~algorithm to construct its supergraph---{\em the~superflow---}from the {\em causal structure only,} without invoking the model parameters.
  As an~application, we show that if all leaves of a~flow are trivial, then the corresponding process produces causal correlations only, i.e., correlations where future data cannot influence past events.
  This strengthens the result that processes, where every directed cycle in their causal structure is chordless, establish causal correlations only.~We~also~discuss~the~main~difficulties~for~the~quantum~generalization~of~the~present~algorithms.
\end{abstract}

\maketitle

\section{Introduction}
A~central barrier for the reconciliation of quantum theory with general relativity, as evocatively pointed out by Hardy~\cite{hardy}, consists in their diametrically opposed concepts: Quantum theory features {\em indefinite quantities\/} that are restricted to evolve along a {\em static causal structure,} general relativity contains only {\em definite\/} quantities, yet the causal structure is {\em dynamic.}
To overcome this barrier, recently much effort has been put into the study of frameworks~\cite{colnaghi,ocb,causalboxes} where the causal connections themselves are governed by quantum theory.
The ``process-matrix framework''~\cite{ocb}, which builds upon the assumption that no deviation from quantum theory is detectable locally, is the most permissive one.
It features the ``quantum switch''~\cite{quantumswitch}, a~prime example of a~process with {\em indefinite causal order.}
In the ``quantum switch,'' a~qubit controls the order in which two quantum gates are applied to a~target system.
That causal order may thus be in an~indefinite state.
This process conceptually exemplifies the {\em general-relativistic trait of dynamical causal order,} as in principle attainable via the displacement of matter, and, simultaneously, the {\em quantum-mechanical trait of superposing these possibilities.}

Causal connections are conveniently expressed and made explicit using {\em causal models\/}~\cite{pearl}, a~prime tool used in statistics, e.g., for causal reasoning:
{\em How may some empirically observed statistics be explained?}
A~causal model consists of two parts:
The causal structure, which is a~directed graph {\em (digraph)} ~$D=(\mathcal V(D),\mathcal E(D)\subseteq\mathcal V(D)^2)$ that encodes cause-effect relations among a~set of variables~$\mathcal V(D)$ via the edges~$\mathcal E(D)$, and the model parameters, which functionally express the respective dependencies (see Fig.~\ref{subfig:causalmodel}).
\begin{figure}
  \subfloat[\label{subfig:causalmodel}]{%
    \begin{tikzpicture}
      \node (A) at (0,0) {$X$};
      \node (B) at (-.5,1) {$Y$};
      \node (C) at (+.5,1) {$Z$};
      \draw[-stealth] (A) -- (B);
      \draw[-stealth] (B) -- (C);
      \draw[-stealth] (A) -- (C);
      \node (param) at (0,+2) {$\{P_X,P_{Y\mid X},P_{Z\mid X,Y}\}$};
    \end{tikzpicture}
  }\hspace{5em}
  \subfloat[\label{subfig:qswitch}]{%
    \begin{tikzpicture}
      \node (P) at (0,0) {$\mathbf P$};
      \node (A) at (-.5,1) {$\mathbf A$};
      \node (B) at (+.5,1) {$\mathbf B$};
      \node (F) at (0,2) {$\mathbf F$};
      \draw[-stealth] (P) -- (A);
      \draw[-stealth] (P) -- (B);
      \draw[-stealth] (A) -- (F);
      \draw[-stealth] (B) -- (F);
      \draw[-stealth] (A) to[out=45,in=180-45] (B);
      \draw[-stealth] (B) to[out=180+45,in=360-45] (A);
      \draw[-stealth] (P) to[out=22.5,in=360-22.5] (F);
    \end{tikzpicture}
  }
  \caption{%
    (a) Causal models: The model parameters with the directed graph constitute a~causal model for the correlations~$P_{X,Y,Z}.$
    (b) Causal structure of the quantum switch:
    The control qubit specified in the global past ($\mathbf P$) influences the global future ($\mathbf F$). Simultaneously, the target system may evolve through~$\mathbf A$ and then through~$\mathbf B$, or vice versa.
    This potentiality of transmitting a~signal from~$\mathbf A$ to~$\mathbf B$ or from~$\mathbf B$ to~$\mathbf A$ is reflected by the directed cycle in the causal structure.
    We use bold letters to indicate that the vertices are split-nodes.
  }
  \label{fig:causalmodels}
\end{figure}
Only recently---also for the sake of aligning Reichenbach's principle~\cite{reichenbach} with Bell correlations~\cite{bell}---causal models where generalized to quantum theory~\cite{fritz,costashrapnel,jmallen,qcm}.
The pioneering work by Costa and Shrapnel~\cite{costashrapnel} is a~first attempt to relate quantum causal models with processes arising from the ``process-matrix framework.''
The latest description of quantum causal models by Barrett, Lorenz, and Oreshkov~\cite{qcm}, which is sufficiently expressive to represent all unitarily extendible processes~\cite{araujo2017}, proved most-useful towards that aim.
Quantum causal models~\cite{qcm} are operational (they inherently allow for interventions) by virtue of following the ``split-node'' approach~\cite{splitnode}:
Each vertex~$v$ of the causal structure represents an~agent with two spaces, an~input space~$\mathcal I_v$, and an~output space~$\mathcal O_v$,
and the agent may perform an~intervention by specifying the state in the output space.
In this operational abstraction, {\em dynamical causal order means that the causal relations are affected by the agent's interventions.}
The description of Barrett, Lorenz, and Oreshkov~\cite{qcm} led to a~series of insights.
For instance, with our collaborator Tselentis~\cite{admissibility}, we show that the admissible quantum causal structures are ``siblings-on-cycles'' graphs; all other graph lead to inconsistencies.
These ``siblings-on-cycles'' graphs are such that every directed cycle therein contains no less than one pair of siblings.\footnote{%
  Two vertices~$u,v$ in a~directed graph~$D$ are {\em siblings,} whenever they share at least one parent, i.e., whenever~$\Pa_D(u)\cap\Pa_D(v)\neq\emptyset$, for~$\Pa_D(\ell):=\{k\mid k\arc\ell\in\mathcal E(D)\}$.
}

\subsection{Problem}
While causal models were extended to incorporate quantum theory, the {\em dynamical\/} aspect of causal relations is only poorly captured.
This dynamical aspect manifests itself only in the presence of directed cycles in the causal structure (see Fig.~\ref{subfig:qswitch}).
These cycles, however, solely reflect the {\em potential\/} transmission of information.
Actually, information traveling along a~cycle {\em must\/} be interrupted before completing the cycle~\cite{admissibility}.
In short, the directed cycles are {\em misleading;} they actually represent only the union of possible causal connections, and no ``causal loop,'' where an~event influences itself~\cite{antinomy}.
This problem is {\em not\/} inherent to quantum theory, but is also present in {\em classical-deterministic\/} causal models.
For example, the ``classical'' version of the quantum switch (where all qubits are replaced by bits) has the same causal structure (Fig.~\ref{subfig:qswitch}).

In the present treatment, we focus on this classical-deterministic case.
This restriction, aside from making the task simpler (the quantum case is at present intractable), is compensated by the following two benefits.
First, it allows us to focus on the dynamical aspect without the risk of having quantum effects appear implicitly.
Second, by studying this restricted case, we effectively study an ``abstraction'' of general relativity.
Note that general relativity is a classical-deterministic theory.
Also, the spacetime structure in general relativity depends on the distribution of matter.
An alternation of that distribution thus induces a change on the causal relations.
Although it is uncommon to describe such alternations in general relativity---after all, the state of affairs are fully determined by the boundary conditions---,
there seems to be no conceptual obstacle in doing so.
Therefore, the ``classical switch'' and any classical-deterministic causal model are, in principle, compatible with general relativity.
In this context, classical-deterministic causal models could be considered as a toy framework of general relativity.

\subsection{Results}
We develop a~tool to capture the {\em evolution of the causal structure in classical-deterministic causal models.}
This tool ``unravels'' the causal structure via the sequential implementation of interventions at source vertices, i.e., parentless vertices.
The resulting {\em flow graph\footnote{%
  A~flow graph is a {\em rooted directed graph:} For every vertex~$v$ there exists a~directed path from the root~$r$ to~$v$.}}
(or {\em flow\/} for short) encodes the evolution of the causal structure, and effectively captures the dynamics induced by the agents' interventions.
This clarifies the nature of the directed cycles in causal models with dynamical causal order.
We also show how to generate a {\em superflow,} an~``outer approximation,'' {\em from the causal structure only, i.e., without invoking the model parameters.}
The superflow is thus a purely qualitative object (independent of the dimensions etc.).
It therefore enables statements of qualitative nature, e.g., the characterization of possible correlations (see below).
This abstraction, however, comes with a price.
For a~given causal model, the superflow may describe additional, potentially unrealizable, evolutions of the causal structure: It may be a supergraph of the flow.
As an~example, Fig.~\ref{fig:flowgraph} shows the flow and the superflow of some causal model.
A specific choice of model parameters, which we present later (Eq.~\eqref{eq:examplemodparam}), results in a tree-like flow (without dashed part).
There, independent of $\mathbf P$'s intervention, the agents~$\mathbf A$ and~$\mathbf B$ are causally connected.
Logically, however, model parameters for the same initial causal structure may exist which produce a flow that comprises the dashed part:
A specific intervention of~$\mathbf P$, which depends on the model parameters, may causally disconnect~$\mathbf A$ and~$\mathbf B$.
\begin{figure}
  \centering
  \begin{tikzpicture}
    \node (P) at (0,0) {$\mathbf P$};
    \node (A) at (-.5,1) {$\mathbf A$};
    \node (B) at (+.5,1) {$\mathbf B$};
    \draw[-stealth] (P) -- (A);
    \draw[-stealth] (P) -- (B);
    \draw[-stealth] (A) to[out=45,in=180-45] (B);
    \draw[-stealth] (B) to[out=180+45,in=360-45] (A);
    \draw (-.75,-.25) rectangle (.75,1.5);
    \def\y{-1}
    \foreach \xoffset/\l in {-2/L,2/R} {
      \node (A\l) at (-.5+\xoffset,\y) {$\mathbf A$};
      \node (B\l) at (+.5+\xoffset,\y) {$\mathbf B$};
      \draw (-.75+\xoffset,-.25+\y) rectangle (.75+\xoffset,.25+\y);
    }
    \draw[-stealth] (AR) -- (BR);
    \draw[-stealth] (BL) -- (AL);
    \node (AC) at (-.5,\y) {$\mathbf A$};
    \node (BC) at (+.5,\y) {$\mathbf B$};
    \draw[dashed] (-.75,-.25+\y) rectangle (.75,.25+\y);
    \def\y{-2}
    \foreach \x/\l in {-1/A,1/B} {
      \node (\l) at (\x,\y) {$\mathbf \l$};
      \draw (-.25+\x,-.25+\y) rectangle (.25+\x,.25+\y);
    }
    \draw[-latex,thick,dashed] (0,-.25) -- (0,-.75);
    \draw[-latex,thick] (0,-.25) -- (-1.5,-.75);
    \draw[-latex,thick] (0,-.25) -- (+1.5,-.75);
    \draw[-latex,thick,dashed] (0,-1.25) -- (-1,-1.75);
    \draw[-latex,thick,dashed] (0,-1.25) -- (+1,-1.75);
    \draw[-latex,thick] (-1.5,-1.25) -- (-1,-1.75);
    \draw[-latex,thick] (+1.5,-1.25) -- (+1,-1.75);
  \end{tikzpicture}
  \caption{%
    Example of a~flow for the model parameters given in Eq.~\eqref{eq:examplemodparam} (without dashed part),
    and the superflow computed by Algorithm~\ref{alg:agnostic} (including dashed part).
  }
  \label{fig:flowgraph}
\end{figure}
The cause-effect relations among the agents for {\em any classical-deterministic causal model with that causal structure\/} are restricted to follow one of the paths in the flow.
The flow of this example has two leaf nodes.
A {\em leaf node,} or {\em leaf,} for short, is a~node without outgoig edges.
Here, both leaves hold trivial graphs, i.e., graphs that are composed by a single vertex.
Consequently, we call such nodes {\em trivial leaves.}
The above observation has stark consequences.
We show that if all leaves of the flow are trivial, then the attainable correlations from the initial causal model decompose causally:
Past data may influence future observations only.
For more general causal structures, this is not necessarily the case~\cite{ocb,baumeler2014,simplest,baumeler2022,admissibility}, and the correlations may be incompatible with any definite causal order.
Clearly, that property also holds for superflows.
With this, we strengthen a~result of Ref.~\cite{admissibility}.
There, we show that the strictly stronger condition of having chordless cycles only is sufficient.
A {\em chord\/} in a~directed cycle~$(v_0\arc\dots\arc v_{\ell-1}\arc v_0)$ is an~edge~$v_i\arc v_j$ with~\mbox{$j\neq i+1\pmod\ell$}.
Finally, we discuss the main challenges in ``unraveling'' quantum causal structures.

\section{Causal models}
In the scenario with the agents~\mbox{${\mathcal V(D)}$},
the respective input and output spaces~$\mathcal I_v$,~$\mathcal O_v$ of each agent~\mbox{$v\in{\mathcal V(D)}$} are finite sets.
An~intervention of agent~$v$ is a~function~\mbox{$\mu_v:\mathcal X_v\times\mathcal I_v\rightarrow\mathcal A_v\times\mathcal O_v$}, where~$\mathcal X_v$ is the set of settings, and~$\mathcal A_v$ the set of results.
We use~$\mathcal M_v$ for the set of all interventions of agent~$v$.
The notation we adopt is as follows.
By~$X_{\mathcal S}$ we denote the natural composition of the objects~$\{X_s\}_{s\in\mathcal S}$, and by~$\underline{X}$ the special case for~$\mathcal S={\mathcal V(D)}$.
The expression~$[\alpha=\beta]$ is one whenever the equality holds, and zero otherwise.
For a~function~$f:\mathcal Y\rightarrow\mathcal Z$, we write~$p^f$ for the induced conditional probability distribution defined through~$p^f(z|y):=[f(y)=z]$.
\begin{definition}[Causal model, consistency, and faithfulness~\cite{qcm,admissibility}]\label{def:causalmodel}
  An~$n$-agent classical-deterministic causal model consists of a~pair~$(D,\mathcal F)$,
  where~$D$ is a~directed graph with~\mbox{$|\mathcal V(D)|=n$} (causal structure),
  and where~$\mathcal F$ is a~family of functions~\mbox{$\{\omega_v:\mathcal O_{\Pa_D(v)}\rightarrow\mathcal I_{v}\}_{v\in{\mathcal V(D)}}$} (model parameters).
  The model is called {\em faithful\/} whenever, for each~$v\in{\mathcal V(D)}$, the function~$\omega_v$ is signaling from each~$\ell\in\Pa_D(v)$ to~$v$:
  \begin{equation}
    \begin{split}
      \forall v\!\in\!\mathcal V(D),\ell\!\in\!\Pa_D(v),\,
      &\exists \hat o\!\in\! \mathcal O_{\Pa_D(v)\setminus\{\ell\}}, (q,r)\!\in\!\mathcal O_\ell^2:
      \\
      \omega_v(q,\hat o) &\neq \omega_v(r,\hat o)
      \,.
    \end{split}
    \label{eq:faithful}
  \end{equation}
  The model is called {\em consistent\/} whenever the contraction of~$\underline{\omega}$ with any choice of interventions~$\underline{\mu}$ is a~function from~$\underline{\mathcal X}$ to~$\underline{\mathcal A}$:
  \begin{equation}
    \begin{split}
      \forall &\underline{\mu}\in\underline{\mathcal M}\,,\exists g\in\{\underline{\mathcal X}\rightarrow\underline{\mathcal A}\}\,,\forall (\underline x,\underline a)\in\underline{\mathcal X}\times\underline{\mathcal A}:\\
      &
      \sum_{(\underline{i},\underline{o})\in\underline{\mathcal I}\times\underline{\mathcal O}}
      [\underline{\omega}(\underline{o}) = \underline{i}]
      [\underline{\mu}(\underline x,\underline i) = (\underline a, \underline o)]
      =
      p^g(\underline a|\underline x)
      \,.
    \end{split}
    \label{eq:consistent}
  \end{equation}
  If not otherwise specified, we use the shorthand term ``causal model'' to refer to a {\em classical-deterministic causal model.}
\end{definition}

A {\em reduction\/} over a~source agent~$s$ is the partial contraction of the causal model with the intervention~$\mu_s$, and where every other intervention is left unspecified.
In the process of the reduction, the source~$s$ also fixes the setting~$x_s$ and records the result~$a_s$.
Here, we will only consider reductions over source agents.

We will make use of the following two central results.
The first states that whenever we reduce over a~source, then the reduced causal model remains consistent.
This result is essentially an~application of Lemma~A.3 of Ref.~\cite{ctc}.
We give a~tailored and simplified proof below.
\begin{theorem}[Reduction\label{thm:reduction}]
  If~$(D,\mathcal F)$ is a~consistent causal model,
  and~$s$ a~source, i.e.,~$\Pa_D(s)=\emptyset$,
  then for any~$o_s\in\mathcal O_s$,
  the reduced causal model~$(D',\mathcal F')$ with
  \begin{align}
    D'
    \!&:=\!
    (\mathcal V(D)\setminus\{s\},\mathcal E(D)\setminus\{ v\arc s,s\arc v\mid v\in\mathcal V(D)\})
    \,,
    \label{eq:subgraph}
    \\
    \mathcal F'
    \!&:=\!
    \{\omega_v\}_{v\in\mathcal V(D')\setminus\Ch_{D}(s)}
    \cup
    \{\omega_v(o_s,\cdots)\}_{v\in\Ch_{D}(s)}
    \,,
    \label{eq:reducedmodparam}
  \end{align}
  where~$\Ch_D(s):=\{v\mid s\arc v\in\mathcal E(D)\}$ are the children of~$s$ in~$D$, and~$\omega_v(o_s,\cdots)$ is the partial application of the function~$\omega_v$,
  is consistent yet possibly unfaithful.
\end{theorem}
\begin{proof}
  The graph~$D'$ is a~causal structure of the reduced causal model because the reduction over~$s$ cannot increase the functional dependencies of the model parameters.
  The reduced causal model remains consistent for the following reason.
  Since~$s$ is parentless, the model parameter~$\omega_s$ is a~constant~$c$.
  Moreover, by the specification of~$\mu_s\in\mathcal M_s$ and~$x_s\in\mathcal X_s$, the output in the space~$\mathcal O_s$ is a~constant~$d$.
  The proof is concluded by plugging in these constants in the consistency condition (Eq.~\eqref{eq:consistent}):
  For any choice of interventions for the remaining parties the contraction with the reduced model parameters co\"{i}ncides with the function~$f'$ defined as the partially evaluated~$f$ where the result~$\mathcal A_s$ is discarded.
\end{proof}
The second result, as highlighted in the introduction, shows that the causal structure of faithful and consistent causal models are severely restrained.
\begin{theorem}[Admissibility~\cite{admissibility}\label{thm:admissibility}]
  If~$(D,\mathcal F)$ is a~faithful and consistent causal model,
  then~$D$ is a~siblings-on-cycles graph.
\end{theorem}

\section{Flow graph}
For the causal model~$(D,\mathcal F)$, the {\em flow graph~$\mathscr F$\/} contains~$D$ as its root (see Fig.~\ref{fig:flowgraph}, and the Appendix for more complex examples), and shows all possible ways in which the causal structure may evolve.
Suppose~$D$ contains a~source vertex~$s$.
This means that the observations made by agent~$s$ cannot be influenced by any other agent; they are independent of the other agents' choice of intervention.
Thus, in order to unravel the causal structure, we may safely implement the intervention of agent~$s$.
This intervention may have an~effect on the causal connections among the remaining agents: different interventions may result in different causal connections.
The flow~$\mathscr F$ we construct captures all possibilities that may arise.
Thus,~$\mathscr F$ branches to all distinct causal structures that may follow after an~intervention at~$s$.
Now, the resulting causal structures may again contain source vertices.
The flow~$\mathscr F$ would then iteratively branch out to causal structures connecting fewer agents.
If there are multiple {\em source\/} vertices, then the flow also branches to the causal structures after intervening at any other source.
Note that branches within the flow might merge (the flow is a rooted directed graph and not necessarily a tree):
The evolution of the causal structure along different choices of interventions may result in the same causal structure.

We give a~construction of the flow by utilizing the model parameters of the causal model, and a~construction of the superflow by ignoring the model parameters.
Note that a superflow may contain paths unrealizable with actual interventions of the agents:
There might exist causal models with the same initial causal structure, but where different parts of the flow are unreachable.

\subsection{Construction: Flow of causal structure}
The algorithm (see Algorithm~\ref{alg:flow}) consists of two parts:
We generate a~rooted directed graph~$\mathscr G$ with causal models as its vertices (lines~1--10), and then we remove all model parameters to arrive at the flow graph~$\mathscr F$ (line 11).
After initiating~$\mathscr G$, we invoke the function~\mbox{\textsc{nontrivialLeavesWithSource}}, which returns all leaves of~$\mathscr G$ whose causal structure is nontrivial and contains at least one source.
Then, we iterate over all leaves, sources vertices therein, and interventions for that source (lines 3--6).
Since~$s$ is a source, the intervention is but the specification of the value on its output space.
We then compute the model parameters of the reduced process (see Eq.~\eqref{eq:reducedmodparam}), and derive the causal structure by determining the functional dependencies.
On line~9, we update~$\mathscr G$ with an edge from that leaf to the reduced one.
We repeat the above as long as~$\mathscr G$ contains nontrivial leaves with source vertices.
After this process terminates, we generate the flow~$\mathscr F$ by discarding the model parameters in all causal models within~$\mathscr G$.
The function~\textsc{removeModelParameters} maps the node-set~$\{(D_i,\mathcal F_i)\}_{i}$ of~$\mathscr G$ to the node-set~$\{D_i\}_i$, and each edge~$(D_i,\mathcal F_i)\arc(D_j,\mathcal F_j)$ to~$D_i\arc D_j$.
\begin{algorithm}[H]
  \caption{Flow (model-parameter aware)}\label{alg:flow}
  \begin{algorithmic}[1]
    \Require causal model $(D,\mathcal F)$
    \Ensure flow $\mathscr F$
    \BeginBox[fill,color=black!5!white]
    \State Initiate~$\mathscr G$ with root~$(D,\mathcal F)$
    \State $\mathcal L \gets \Call{nontrivialLeavesWithSource}{\mathscr G}$
    \While{$\mathcal L \neq \emptyset$}
      \For{$(L,\mathcal F_L)\in\mathcal L$}
        \For{$s\in\Call{sourceVertices}{L}$}
          \For{$o_s\in\mathcal O_s$}
            \State $\mathcal F_L' \gets \Call{reduceModParam}{\mathcal F_L,L,s,o_s}$
            \State $L' \gets \Call{deriveCausalStructure}{\mathcal F_L'}$
            \State $\mathscr G \gets \mathscr G \cup \{(L,\mathcal F_L)\arc(L',\mathcal F_L')\}$
          \EndFor
        \EndFor
      \EndFor
      \State $\mathcal L \gets \Call{nontrivialLeavesWithSource}{\mathscr G}$
    \EndWhile
    \EndBox
    \BeginBox[fill,color=black!5!white]
    \State $\mathscr F \gets \Call{removeModelParameters}{\mathscr G}$
    \EndBox
    \State \Return $\mathscr F$
  \end{algorithmic}
\end{algorithm}

\subsubsection{Correctness}
The algorithm terminates after finitely many steps because
the initial causal model has a~finite number of agents, and the reduction~(lines 7--8) removes one agent at-a-time.
{\em Consistency\/} follows from Theorem~\ref{thm:reduction}, and {\em faithfulness\/} by construction (line 8).
Finally, the iteration on line~6 covers all interventions of agent~$s$.
With this, the flow~$\mathscr F$ describes all possible paths in which the causal structure may evolve.

\subsubsection{Example}
Take the causal structure shown in the root of Fig.~\ref{fig:flowgraph}, and amend it with the model parameters
\begin{align}
  \omega_A(o_P,o_B):=o_Po_B
  \,,
  \quad
  \omega_B(o_P,o_A):=(1-o_P)o_A
  \label{eq:examplemodparam}
  \,,
\end{align}
where all input and all output spaces, but~$\mathcal I_P$, are~$\{0,1\}$, and where the input space to~$P$ is trivial.
The resulting flow~$\mathscr F$ is shown in Fig.~\ref{fig:flowgraph}.

\subsection{Construction: Superflow of causal structure}
For this construction (see Algorithm~\ref{alg:agnostic}), we make use of Theorem~\ref{thm:reduction} in conjunction with Theorem~\ref{thm:admissibility}:
Whenever we remove a~source from a~consistent causal model, then the resulting causal structure must be a~siblings-on-cycles graph.
The algorithm is as follows.
After initiating the graph~$\mathscr S$ with the causal structure~$D$, we iterate over all leaves that contain source vertices, and over each source (lines 4--5).
Then we define~$L'$ by removing the source~$s$ from the graph~$L$, i.e., we remove~$s$ from the set of vertices, and all edges starting from or ending in~$s$ (see also Eq.~\eqref{eq:subgraph}).
After that, we collect all edges of~$L'$ that point to any child of~$s$.
Then we iterate over all possible subsets, and remove these edges from~$L'$.
If the resulting graph is a~siblings-on-cycles graph, as decided by the function \textsc{isSOC}, we add it to the superflow~$\mathscr S$.
\begin{algorithm}[H]
  \caption{Superflow (model-parameter agnostic)}\label{alg:agnostic}
  \begin{algorithmic}[1]
    \Require causal structure $D$
    \Ensure superflow $\mathscr S$
    \State Initiate~$\mathscr S$ with root~$D$
    \State $\mathcal L \gets \Call{nontrivialLeavesWithSource}{\mathscr S}$
    \While{$\mathcal L \neq \emptyset$}
      \For{$L\in\mathcal L$}
        \For{$s\in\Call{sourceVertices}{L}$}
          \State $L' \gets L \setminus \{s\}$
          \State $\mathcal E \gets \{v\arc k \mid k\in\Ch_L(s),v\in\Pa_{L'}(k)\}$
          \For{$\mathcal R \subseteq \mathcal E$}
            \State $L'' \gets L'\setminus \mathcal R$
            \If{\Call{isSOC}{$L''$}}
              \State $\mathscr S \gets \mathscr S \cup \{L\arc L''\}$
            \EndIf
          \EndFor
        \EndFor
      \EndFor
      \State $\mathcal L \gets \Call{nontrivialLeavesWithSource}{\mathscr S}$
    \EndWhile
    \State \Return $\mathscr S$
  \end{algorithmic}
\end{algorithm}

\subsubsection{Correctness}
For any~$L$ on line~4 we will reach line~11 {\em at least once.}
This together with the finiteness of~$D$ implies that Algorithm~\ref{alg:agnostic} terminates after finitely many steps.
For a~causal model~$(D,\mathcal F)$, let~$\mathscr F$ be the flow computed using Algorithm~\ref{alg:flow}, and~$\mathscr S$ the superflow using Algorithm~\ref{alg:agnostic}.
Now, any path~$\pi=D\arc{D_2}\arc\dots\arc{D_\ell}$ in flow~$\mathscr F$ is also present in~$\mathscr S$.
The reason for this is that any reduction over a~source~$s$ may alter the model parameters of the agents~$\Ch_L(s)$ only (see line~7 of Algorithm~\ref{alg:flow} and Eq.~\eqref{eq:reducedmodparam});
and Algorithm~\ref{alg:agnostic} exhausts all possibilities.

\subsection{Runtime}
For simplicity, suppose that the maximum cardinality of the output spaces,~$\max |\mathcal O_v|$, is a constant (this is only relevant for the runtime analysis of Algorithm~\ref{alg:flow}).
Both algorithms presented have an exponential runtime in the number of vertices~$n$ of the causal structure in the root.
In particular, the runtime is
\begin{equation}
  O(2^{\text{poly}(n)})
  \,,
  \label{eq:runtime}
\end{equation}
where~$\text{poly}(n)$ is some polynomial in~$n$.
To see this, first note that a flow could coincide with the superflow.
The number of potential nodes~$N$ in the (super)flow is exponential in~$n$.
So, in Algorithm~\ref{alg:flow} we repeat the lines~5--9 at most~$N$ times, and in Algorithm~\ref{alg:agnostic}, we repeat the lines~8--11 at most~$N$ times.
In both algorithms, each of these repetitions can be done in at most an exponential number of steps.

We give more details on showing that~$N$ is exponential in~$n$.
Towards that, let~$L$ be a node in the (super)flow, and assume that the digraph~$L$ has a single source vertex~$s$.
The number of outgoing edges from~$s$ is~$O(2^{n^2})$.
The reason for this is that for each parent vertex of~$s$, we consider all combinations of possibly ingoing edges (lines 7--8 in Algorithm~\ref{alg:agnostic}).
For each parent of~$s$, there are~$O(2^n)$ such possibilities, and there are~$O(n)$ parents of~$s$.
In general,~$L$ has~$O(n)$ source vertices, so we get another exponent of~$n$, and we remain in the form of Eq.~\eqref{eq:runtime}.

\section{Correlations}
The correlations among the agents~$\mathcal S$ are expressed by the conditional probabilities~$p(a_{\mathcal S}|x_{\mathcal S})$ of observing~$a_{\mathcal S}$ under the settings~$x_{\mathcal S}$.
Such correlations are termed {\em causal\/}~\cite{ocb,og,multi2} if and only if they decompose as
\begin{equation}
  p(a_{\mathcal S}|x_{\mathcal S})
  =
  \sum_{v\in\mathcal S}
  p_{\mathcal S}(v)
  p_v(a_v|x_v)
  p^{x_v}_{a_v}(a_{\mathcal S\setminus\{v\}}|x_{\mathcal S\setminus\{v\}})
  \,,
\end{equation}
where~$p_{\mathcal S}$ is a~probability distribution over~$\mathcal S$,~$p_v$ a~conditional probability distribution of agent~$v$ observing~$a_v$ under the setting~$x_v$,
and where~$p^{x_v}_{a_v}(a_{\mathcal S\setminus\{v\}}|x_{\mathcal S\setminus\{v\}})$ are causal correlations again.
Causal correlations are compatible with a dynamical specification of the causal order, where each agent may influence the order of the agents within its future.
Single-agent correlations are trivially causal.
For a~causal model~$(D,\mathcal F)$ and fixed interventions~$\underline\mu\in\underline{\mathcal M}$, the obtained correlations are~$p^g(\underline a|\underline x)$, where~$g$ is the respective function in Eq.~\eqref{eq:consistent}.
As an application of the developed tool, we present the following result.
\begin{theorem}[Causal correlations\label{thm:cc}]
  If all leaves of a~flow~$\mathscr F$ are trivial,
  then the correlations attainable from a~causal model~$(D,\mathcal F)$ with flow~$\mathscr F$ are causal.
\end{theorem}
\begin{proof}
  For any choice of interventions of the agents, there exists an~ordering in which to perform the contractions, such that in every iteration a {\em source\/} is contracted.
  The correlations~$p^g(\underline a|\underline x)$ decompose in that order.
\end{proof}

This theorem trivially also holds for {\em superflows.}
Using this theorem we can identify causal structures (digraphs) that, when extended to a causal model, always produce causal correlations.
So, this theorem strengthens the result~\cite{admissibility} that the correlations from a~causal model with {\em chordless\/} cycles only in its causal structure must be causal.
Examples of causal structures with chordal cycles, but that---according to Theorem~\ref{thm:cc}---still produce causal correlations only, are given in Fig.~\ref{fig:examplegraph}.
\begin{figure}
  \def\dist{0.7}
  \centering
  \subfloat[\label{subfig:exn4}]{%
    \begin{tikzpicture}
      \node[vertex] (P) at (0,-\dist) {};
      \node[vertex] (A) at (-\dist,0) {};
      \node[vertex] (B) at (+\dist,0) {};
      \node[vertex] (F) at (0,+\dist) {};
      \draw[-stealth] (P) -- (A);
      \draw[-stealth] (P) -- (B);
      \draw[-stealth] (A) to[out=20,in=180-20] (B);
      \draw[-stealth] (B) to[out=180+20,in=360-20] (A);
      \draw[-stealth] (A) -- (F);
      \draw[-stealth] (F) -- (B);
    \end{tikzpicture}
  }
  \hfill
  \subfloat[\label{subfig:exn5a}]{%
    \begin{tikzpicture}
      \node[vertex] (A) at (0,-\dist) {};
      \node[vertex] (B) at (-\dist,-0.5*\dist) {};
      \node[vertex] (C) at (0,0) {};
      \node[vertex] (D) at (-\dist,0.5*\dist) {};
      \node[vertex] (E) at (0,\dist) {};
      \draw[-stealth] (A) -- (B);
      \draw[-stealth] (B) -- (D);
      \draw[-stealth] (D) -- (E);
      \draw[-stealth] (E) -- (C);
      \draw[-stealth] (C) -- (B);
      \draw[-stealth] (A) -- (C);
      \draw[-stealth] (D) -- (C);
    \end{tikzpicture}
  }
  \hfill
  \subfloat[\label{subfig:exn5b}]{%
    \begin{tikzpicture}
      \def\dist{0.5}
      \node[vertex] (A) at (0,-\dist) {};
      \node[vertex] (B) at (-\dist,0) {};
      \node[vertex] (C) at (\dist,0) {};
      \node[vertex] (D) at (-2*\dist,\dist) {};
      \node[vertex] (E) at (2*\dist,\dist) {};
      \draw[-stealth] (A) -- (B);
      \draw[-stealth] (A) -- (C);
      \draw[-stealth] (B) -- (C);
      \draw[-stealth] (B) -- (D);
      \draw[-stealth] (D) -- (C);
      \draw[-stealth] (C) -- (E);
      \draw[-stealth] (E) -- (B);
    \end{tikzpicture}
  }
  \hfill
  \subfloat[\label{subfig:exn6}]{%
    \begin{tikzpicture}
      \def\dist{0.5}
      \node[vertex] (A) at (0,-\dist) {};
      \node[vertex] (B) at (-\dist,0) {};
      \node[vertex] (C) at (\dist,0) {};
      \node[vertex] (D) at (-2*\dist,\dist) {};
      \node[vertex] (E) at (2*\dist,\dist) {};
      \node[vertex] (F) at (0,\dist) {};
      \draw[-stealth] (A) -- (B);
      \draw[-stealth] (A) -- (C);
      \draw[-stealth] (B) -- (C);
      \draw[-stealth] (C) -- (E);
      \draw[-stealth] (C) -- (F);
      \draw[-stealth] (D) -- (B);
      \draw[-stealth] (E) -- (F);
      \draw[-stealth] (F) -- (B);
      \draw[-stealth] (F) -- (D);
    \end{tikzpicture}
  }
  \caption{%
    These causal structures give causal correlations only.
  }
  \label{fig:examplegraph}
\end{figure}
In Fig.~\ref{fig:superflow4b} (Appendix), we give the superflow of the causal structure shown in Fig.~\ref{subfig:exn4}.
In the Appendix we also list all connected four-vertex digraphs in the ``gap'' between the result of Ref.~\cite{admissibility} and Theorem~\ref{thm:cc}.

\section{Quantum generalization}
In the quantum case~\cite{qcm}, the input and output space of the agents are Hilbert spaces.
Also, the model parameters are a~family of completely positive trace-preserving maps as opposed to functions, that, in the Choi representation~\cite{choi}, pairwise commute.
This ensures that the quantum process~$W:=\prod_{v\in\mathcal V(D)}\omega_{v}$ is well-defined.
Finally, faithfulness (see Eq.~\eqref{eq:faithful}) is defined analogously, and consistency (see Eq.~\eqref{eq:consistent}) is
\begin{equation}
  \forall \{\mu_v\in\cptp_v\}_{v\in\mathcal V(D)}:
  \Tr\left[
    W
    \cdot
    \bigotimes_{v\in\mathcal V(D)}
    \mu_v
  \right]
  =
  1
  \,,
\end{equation}
where~$\cptp_v$ is the set of completely positive trace-preserving maps from~$\mathcal L(\mathcal I_v)$ (the set of linear operators on~$\mathcal I_v$) to~$\mathcal L(\mathcal O_v)$ in the Choi representation.

Take a~quantum causal model~$(D,\mathcal F)$ where agent~$s$ is a~source, and suppose the input Hilbert space~$\mathcal I_s$ is trivial.
If agent~$s$ performs as intervention the state preparation~$\mu_s$, then it is unclear how to define the {\em reduced\/} model parameters, and whether they may be defined in a~meaningful way.
The reduced process~$W'$ is given by~$\Tr_{\mathcal O_s}[W\mu_s]$.
Generally, this partial trace does not extend to~$\prod_{v\in\mathcal V(D')\setminus\Ch_{D}(s)} \omega_{v} \cdot \prod_{v\in\Ch_{D}(s)} \Tr_{\mathcal O_s}[\omega_v\mu_s]$,
for~$D'$ defined as in Eq.~\eqref{eq:subgraph}.

\section{Conclusion}
We present a tool---the flow of causal structures---for studying the dynamical aspect of causal models.
This tool ``linearizes'' cycles in causal structures, and hence allows for conventional step-by-step analysis of communication schemes, computation, and general information processing under {\em dynamical causal order.}
In fact, the flow of the directed path graph (a directed line) is just the directed path graph again.
We demonstrate the usability of this tool by invoking it in the context of correlations: We show that if all leaves of the flow are trivial, then the possible correlations from the causal model must decompose causally.
Possible further applications of this tool are as follows.
Flows might be beneficial in bounding the computational power of classical and quantum computation without causal order~\cite{tameness,postbqp}.
Also, they might be helpful in resolving the present conjecture~\cite{admissibility} that each siblings-on-cycles graph is part of a~consistent and faithful causal model.
By design, flows allow us to monitor the stream of information in dynamical contexts.
Arguably, classical-deterministic processes with dynamical causal order are in principle implementable within general relativity.
Thus, on a foundations-of-physics front, flows might facilitate an information-based treatment of general relativity,
and the design of relativistic protocols for information processing.
In a related work, Vilasini and Renner~\cite{vvrr} show the limitations of embedding quantum cyclic causal models within a {\em fixed\/} relativistic spacetime.
Here, we take a different approach and study relativistic causal structures which may be {\em dynamic,} as guided by the principles of general relativity.
We leave open the embedding of quantum cyclic causal models in such {\em dynamic\/} spacetimes.

\noindent
{\bf Code Availability.}
We have implemented Algorithm~\ref{alg:agnostic} as a \texttt{C} program~\cite{code}.
For the Appendix, we also made use of the ``SOC Observation Code''~\cite{codesoc}.

\noindent
{\bf Acknowledgments.}
We thank Lefteris Tselentis for helpful discussions, and two anonymous referees for their valuable input.
This work is supported by the Swiss National Science Foundation (SNF) through project~214808, and by the Hasler Foundation through project~24010.

\bibliography{references.bib}
\onecolumngrid
\newpage
\appendix
\section*{Superflow examples}
Figure~\ref{fig:n4all} gives a comprehensive list of all connected four-vertex digraphs that are of relevance for Theorem~\ref{thm:cc}.
These digraphs satisfy some properties.
First, they are admissible as causal structures, i.e., they are siblings-on-cycles graphs.
Second, each of them has at least one source vertex.
And third, each of them contains at least one directed cycle with a chord.
As by our algorithm, it turns out that all digraphs but (a), (d), and (g), produce causal correlations only.

\begin{figure}[H]
  \def\dist{0.7}
  \def\nodes{%
    \node[vertex] (P) at (0,-\dist) {};
    \node[vertex] (A) at (-\dist,0) {};
    \node[vertex] (B) at (+\dist,0) {};
    \node[vertex] (F) at (0,+\dist) {};
  }
  \centering
  \subfloat[\label{subfig:n4ga}]{%
    \begin{tikzpicture}
      \nodes
      \draw[-stealth] (P) -- (A);
      \draw[-stealth] (A) to[out=20,in=180-20] (B);
      \draw[-stealth] (B) to[out=180+20,in=360-20] (A);
      \draw[-stealth] (A) to[out=45+20,in=180+45-20] (F);
      \draw[-stealth] (F) to[out=180+45+20,in=360+45-20] (A);
      \draw[-stealth] (F) to[out=-45+20,in=180-45-20] (B);
      \draw[-stealth] (B) to[out=180+-45+20,in=360-45-20] (F);
    \end{tikzpicture}
  }
  \hfill
  \subfloat[\label{subfig:n4gb}]{%
    \begin{tikzpicture}
      \nodes
      \draw[-stealth] (P) -- (A);
      \draw[-stealth] (P) -- (B);
      \draw[-stealth] (A) to[out=20,in=180-20] (B);
      \draw[-stealth] (B) to[out=180+20,in=360-20] (A);
      \draw[-stealth] (A) -- (F);
      \draw[-stealth] (F) -- (B);
    \end{tikzpicture}
  }
  \hfill
  \subfloat[\label{subfig:n4gc}]{%
    \begin{tikzpicture}
      \nodes
      \draw[-stealth] (P) -- (A);
      \draw[-stealth] (P) -- (B);
      \draw[-stealth] (A) to[out=20,in=180-20] (B);
      \draw[-stealth] (B) to[out=180+20,in=360-20] (A);
      \draw[-stealth] (A) -- (F);
      \draw[-stealth] (F) to[out=-45+20,in=180-45-20] (B);
      \draw[-stealth] (B) to[out=180+-45+20,in=360-45-20] (F);
    \end{tikzpicture}
  }
  \hfill
  \subfloat[\label{subfig:n4gd}]{%
    \begin{tikzpicture}
      \nodes
      \draw[-stealth] (P) -- (A);
      \draw[-stealth] (P) -- (B);
      \draw[-stealth] (A) to[out=20,in=180-20] (B);
      \draw[-stealth] (B) to[out=180+20,in=360-20] (A);
      \draw[-stealth] (A) to[out=45+20,in=180+45-20] (F);
      \draw[-stealth] (F) to[out=180+45+20,in=360+45-20] (A);
      \draw[-stealth] (F) to[out=-45+20,in=180-45-20] (B);
      \draw[-stealth] (B) to[out=180+-45+20,in=360-45-20] (F);
    \end{tikzpicture}
  }
  \hfill
  \subfloat[\label{subfig:n4ge}]{%
    \begin{tikzpicture}
      \nodes
      \draw[-stealth] (P) -- (A);
      \draw[-stealth] (P) -- (B);
      \draw[-stealth] (P) -- (F);
      \draw[-stealth] (B) -- (A);
      \draw[-stealth] (A) to[out=45+20,in=180+45-20] (F);
      \draw[-stealth] (F) to[out=180+45+20,in=360+45-20] (A);
      \draw[-stealth] (F) -- (B);
    \end{tikzpicture}
  }
  \hfill
  \subfloat[\label{subfig:n4gf}]{%
    \begin{tikzpicture}
      \nodes
      \draw[-stealth] (P) -- (A);
      \draw[-stealth] (P) -- (B);
      \draw[-stealth] (P) -- (F);
      \draw[-stealth] (B) -- (A);
      \draw[-stealth] (A) to[out=45+20,in=180+45-20] (F);
      \draw[-stealth] (F) to[out=180+45+20,in=360+45-20] (A);
      \draw[-stealth] (F) to[out=-45+20,in=180-45-20] (B);
      \draw[-stealth] (B) to[out=180+-45+20,in=360-45-20] (F);
    \end{tikzpicture}
  }
  \hfill
  \subfloat[\label{subfig:n4gg}]{%
    \begin{tikzpicture}
      \nodes
      \draw[-stealth] (P) -- (A);
      \draw[-stealth] (P) -- (B);
      \draw[-stealth] (P) -- (F);
      \draw[-stealth] (A) to[out=20,in=180-20] (B);
      \draw[-stealth] (B) to[out=180+20,in=360-20] (A);
      \draw[-stealth] (A) to[out=45+20,in=180+45-20] (F);
      \draw[-stealth] (F) to[out=180+45+20,in=360+45-20] (A);
      \draw[-stealth] (F) to[out=-45+20,in=180-45-20] (B);
      \draw[-stealth] (B) to[out=180+-45+20,in=360-45-20] (F);
    \end{tikzpicture}
  }
  \caption{%
    All connected four-node cyclic digraphs with at lest one source and at least one cycle with a chord.
  }
  \label{fig:n4all}
\end{figure}

In Figs.~\ref{fig:superflow4b} and~\ref{fig:exampleNEW}, we present the superflows of the causal structures given in Figs.~\ref{subfig:n4gb} and~\ref{subfig:n4gc}.
In both cases, the superflow has trivial leaves.
By Theorem~\ref{thm:cc}, the correlations produced by any causal model with such a causal structure are restricted to decompose causally.
In Fig.~\ref{fig:exampleNEW2}, we present the superflow of the causal structure of Fig.~\ref{subfig:n4gg}.
Here, we have a non-trivial leaf:
There might exist model parameters such that the correlations produced by such a causal structure do not decompose causally.

\newcommand{\sfpad}{0.25}
\newcommand{\sfdeg}{10}
\begin{figure}[H]
  \centering
  \begin{tikzpicture}
    \node (P) at (0,0) {$\mathbf P$};
    \node (A) at (-1,1) {$\mathbf A$};
    \node (B) at (+1,1) {$\mathbf B$};
    \node (C) at (0,2) {$\mathbf C$};
    \draw[-stealth] (P) -- (A);
    \draw[-stealth] (A) -- (C);
    \draw[-stealth] (P) -- (B);
    \draw[-stealth] (C) -- (B);
    \draw[-stealth] (A) to[out=\sfdeg,in=180-\sfdeg] (B);
    \draw[-stealth] (B) to[out=180+\sfdeg,in=-\sfdeg] (A);
    \draw (-1-\sfpad,0-\sfpad) rectangle (1+\sfpad,2+\sfpad);
    \def\y{-3}
    \foreach \xoffset/\l in {-6/LL,-3/L,0/C,3/R,6/RR} {
      \node (A\l) at (\xoffset-1,\y) {$\mathbf A$};
      \node (B\l) at (\xoffset+1,\y) {$\mathbf B$};
      \node (C\l) at (\xoffset+0,\y+1) {$\mathbf C$};
      \draw[dashed] (\xoffset-1-\sfpad,\y-\sfpad) rectangle (\xoffset+1+\sfpad,\y+1+\sfpad);
      \draw[dashed,-latex,thick] (0,0-\sfpad) -- (\xoffset,\y+1+\sfpad);
    }
    % LL
    \draw[-stealth] (BLL) -- (ALL);
    \draw[-stealth] (ALL) -- (CLL);
    % L
    \draw[-stealth] (AL) -- (CL);
    % C
    \draw[-stealth] (AC) -- (BC);
    \draw[-stealth] (AC) -- (CC);
    % R
    \draw[-stealth] (AR) -- (CR);
    \draw[-stealth] (CR) -- (BR);
    \draw[-stealth] (AR) -- (BR);
    % R
    \draw[-stealth] (ARR) -- (CRR);
    \draw[-stealth] (CRR) -- (BRR);
    \def\y{-5}
    \foreach \xoffset/\l/\A/\B in {-2/NL/A/C,0/N/C/B,2/NR/C/B} {
      \node (\A\l) at (\xoffset-0.5,\y) {$\mathbf \A$};
      \node (\B\l) at (\xoffset+0.5,\y) {$\mathbf \B$};
      \draw[dashed] (\xoffset-0.5-\sfpad,\y-\sfpad) rectangle (\xoffset+0.5+\sfpad,\y+\sfpad);
    }
    % NL
    \draw[-stealth] (ANL) -- (CNL);
    \draw[dashed,-latex,thick] ($ (ALL) + (1,-\sfpad) $) -- ($ (ANL) + (0.5,\sfpad) $);
    \draw[dashed,-latex,thick] ($ (AL) + (1,-\sfpad) $) -- ($ (ANL) + (0.5,\sfpad) $);
    % N
    \draw[dashed,-latex,thick] ($ (AL) + (1,-\sfpad) $) -- ($ (CN) + (0.5,\sfpad) $);
    \draw[dashed,-latex,thick] ($ (AC) + (1,-\sfpad) $) -- ($ (CN) + (0.5,\sfpad) $);
    \draw[dashed,-latex,thick] ($ (AR) + (1,-\sfpad) $) -- ($ (CN) + (0.5,\sfpad) $);
    % NR
    \draw[-stealth] (CNR) -- (BNR);
    \draw[dashed,-latex,thick] ($ (AR) + (1,-\sfpad) $) -- ($ (CNR) + (0.5,\sfpad) $);
    \draw[dashed,-latex,thick] ($ (ARR) + (1,-\sfpad) $) -- ($ (CNR) + (0.5,\sfpad) $);
    \def\y{-7}
    \foreach \xoffset/\l/\A in {-0.75/BotL/C,0.75/BotR/B} {
      \node (\A\l) at (\xoffset,\y) {$\mathbf \A$};
      \draw (\xoffset-\sfpad,\y-\sfpad) rectangle (\xoffset+\sfpad,\y+\sfpad);
    }
    \draw[dashed,-latex,thick] ($ (ANL) + (0.5,-\sfpad) $) -- ($ (CBotL) + (0,\sfpad) $);
    \draw[dashed,-latex,thick] ($ (CN) + (0.5,-\sfpad) $) -- ($ (CBotL) + (0,\sfpad) $);
    \draw[dashed,-latex,thick] ($ (CN) + (0.5,-\sfpad) $) -- ($ (BBotR) + (0,\sfpad) $);
    \draw[dashed,-latex,thick] ($ (CNR) + (0.5,-\sfpad) $) -- ($ (BBotR) + (0,\sfpad) $);
  \end{tikzpicture}
  \caption{%
    Superflow of the example given in Fig.~\ref{subfig:exn4} (Fig.~\ref{subfig:n4gb} in the Appendix).
  }
  \label{fig:superflow4b}
\end{figure}
\begin{figure}[H]
  \centering
  \begin{tikzpicture}
    \node (P) at (0,0) {$\mathbf P$};
    \node (A) at (-1,1) {$\mathbf A$};
    \node (B) at (+1,1) {$\mathbf B$};
    \node (C) at (0,2) {$\mathbf C$};
    \draw[-stealth] (P) -- (A);
    \draw[-stealth] (A) -- (C);
    \draw[-stealth] (P) -- (B);
    \draw[-stealth] (B) to[out=90+45+\sfdeg,in=180+90+45-\sfdeg] (C);
    \draw[-stealth] (C) to[out=180+90+45+\sfdeg,in=90+45-\sfdeg] (B);
    \draw[-stealth] (A) to[out=\sfdeg,in=180-\sfdeg] (B);
    \draw[-stealth] (B) to[out=180+\sfdeg,in=-\sfdeg] (A);
    \draw (-1-\sfpad,0-\sfpad) rectangle (1+\sfpad,2+\sfpad);
    \def\y{-3}
    \foreach \xoffset/\l in {-4.5/LL,-1.5/L,1.5/R,4.5/RR} {
      \node (A\l) at (\xoffset-1,\y) {$\mathbf A$};
      \node (B\l) at (\xoffset+1,\y) {$\mathbf B$};
      \node (C\l) at (\xoffset+0,\y+1) {$\mathbf C$};
      \draw[dashed] (\xoffset-1-\sfpad,\y-\sfpad) rectangle (\xoffset+1+\sfpad,\y+1+\sfpad);
      \draw[dashed,-latex,thick] (0,0-\sfpad) -- (\xoffset,\y+1+\sfpad);
    }
    % LL
    \draw[-stealth] (ALL) -- (CLL);
    \draw[-stealth] (BLL) -- (CLL);
    \draw[-stealth] (BLL) -- (ALL);
    % L
    \draw[-stealth] (AL) -- (CL);
    \draw[-stealth] (BL) -- (CL);
    % R
    \draw[-stealth] (AR) -- (CR);
    \draw[-stealth] (BR) -- (CR);
    \draw[-stealth] (AR) -- (BR);
    % R
    \draw[-stealth] (ARR) -- (CRR);
    \draw[-stealth] (ARR) -- (BRR);
    \draw[-stealth] (BRR) to[out=90+45+\sfdeg,in=180+90+45-\sfdeg] (CRR);
    \draw[-stealth] (CRR) to[out=180+90+45+\sfdeg,in=90+45-\sfdeg] (BRR);
    \def\y{-5}
    \foreach \xoffset/\l/\A/\B in {-4/NLL/A/C,-2/NL/A/C,0/N/C/B,2/NR/C/B,4/NRR/C/B} {
      \node (\A\l) at (\xoffset-0.5,\y) {$\mathbf \A$};
      \node (\B\l) at (\xoffset+0.5,\y) {$\mathbf \B$};
      \draw[dashed] (\xoffset-0.5-\sfpad,\y-\sfpad) rectangle (\xoffset+0.5+\sfpad,\y+\sfpad);
    }
    % NLL
    \draw[dashed,-latex,thick] ($ (ALL) + (1,-\sfpad) $) -- ($ (ANLL) + (0.5,\sfpad) $);
    \draw[dashed,-latex,thick] ($ (AL) + (1,-\sfpad) $) -- ($ (ANLL) + (0.5,\sfpad) $);
    % NL
    \draw[-stealth] (ANL) -- (CNL);
    \draw[dashed,-latex,thick] ($ (ALL) + (1,-\sfpad) $) -- ($ (ANL) + (0.5,\sfpad) $);
    \draw[dashed,-latex,thick] ($ (AL) + (1,-\sfpad) $) -- ($ (ANL) + (0.5,\sfpad) $);
    % N
    \draw[-stealth] (BN) -- (CN);
    \draw[dashed,-latex,thick] ($ (AL) + (1,-\sfpad) $) -- ($ (CN) + (0.5,\sfpad) $);
    \draw[dashed,-latex,thick] ($ (AR) + (1,-\sfpad) $) -- ($ (CN) + (0.5,\sfpad) $);
    \draw[dashed,-latex,thick] ($ (ARR) + (1,-\sfpad) $) -- ($ (CN) + (0.5,\sfpad) $);
    % NR
    \draw[dashed,-latex,thick] ($ (AL) + (1,-\sfpad) $) -- ($ (CNR) + (0.5,\sfpad) $);
    \draw[dashed,-latex,thick] ($ (AR) + (1,-\sfpad) $) -- ($ (CNR) + (0.5,\sfpad) $);
    \draw[dashed,-latex,thick] ($ (ARR) + (1,-\sfpad) $) -- ($ (CNR) + (0.5,\sfpad) $);
    % NRR
    \draw[-stealth] (CNRR) -- (BNRR);
    \draw[dashed,-latex,thick] ($ (ARR) + (1,-\sfpad) $) -- ($ (CNRR) + (0.5,\sfpad) $);
    \def\y{-7}
    \foreach \xoffset/\l/\A in {-1.5/BotL/A,0/BotN/C,1.5/BotR/B} {
      \node (\A\l) at (\xoffset,\y) {$\mathbf \A$};
      \draw (\xoffset-\sfpad,\y-\sfpad) rectangle (\xoffset+\sfpad,\y+\sfpad);
    }
    \draw[dashed,-latex,thick] ($ (ANLL) + (0.5,-\sfpad) $) -- ($ (ABotL) + (0,\sfpad) $);
    \draw[dashed,-latex,thick] ($ (ANLL) + (0.5,-\sfpad) $) -- ($ (CBotN) + (0,\sfpad) $);
    \draw[dashed,-latex,thick] ($ (ANL) + (0.5,-\sfpad) $) -- ($ (CBotN) + (0,\sfpad) $);
    \draw[dashed,-latex,thick] ($ (CN) + (0.5,-\sfpad) $) -- ($ (CBotN) + (0,\sfpad) $);
    \draw[dashed,-latex,thick] ($ (CNR) + (0.5,-\sfpad) $) -- ($ (CBotN) + (0,\sfpad) $);
    \draw[dashed,-latex,thick] ($ (CNR) + (0.5,-\sfpad) $) -- ($ (BBotR) + (0,\sfpad) $);
    \draw[dashed,-latex,thick] ($ (CNRR) + (0.5,-\sfpad) $) -- ($ (BBotR) + (0,\sfpad) $);
  \end{tikzpicture}
  \caption{%
    Superflow of the example given in Fig.~\ref{subfig:n4gc}.
  }
  \label{fig:exampleNEW}
\end{figure}
\vfill
\begin{figure}[H]
  \centering
  \begin{tikzpicture}
    \node (P) at (0,0) {$\mathbf P$};
    \node (A) at (-1,1) {$\mathbf A$};
    \node (B) at (+1,1) {$\mathbf B$};
    \node (C) at (0,2) {$\mathbf C$};
    \draw[-stealth] (P) -- (A);
    \draw[-stealth] (P) -- (B);
    \draw[-stealth] (A) to[out=45+\sfdeg,in=180+45-\sfdeg] (C);
    \draw[-stealth] (C) to[out=180+45+\sfdeg,in=45-\sfdeg] (A);
    \draw[-stealth] (B) to[out=90+45+\sfdeg,in=180+90+45-\sfdeg] (C);
    \draw[-stealth] (C) to[out=180+90+45+\sfdeg,in=90+45-\sfdeg] (B);
    \draw[-stealth] (A) to[out=\sfdeg,in=180-\sfdeg] (B);
    \draw[-stealth] (B) to[out=180+\sfdeg,in=-\sfdeg] (A);
    \draw (-1-\sfpad,0-\sfpad) rectangle (1+\sfpad,2+\sfpad);
    \def\y{-3}
    \foreach \xoffset/\l in {+4.5/5, +1.5/4, -1.5/1, -4.5/2, -7.5/3} {
      \node (A\l) at (\xoffset-1,\y) {$\mathbf A$};
      \node (B\l) at (\xoffset+1,\y) {$\mathbf B$};
      \node (C\l) at (\xoffset+0,\y+1) {$\mathbf C$};
      \draw[dashed] (\xoffset-1-\sfpad,\y-\sfpad) rectangle (\xoffset+1+\sfpad,\y+1+\sfpad);
      \draw[dashed,-latex,thick] (0,0-\sfpad) -- (\xoffset,\y+1+\sfpad);
    }
    \foreach \xoffset/\l in {+7.5/6} {
      \node (A\l) at (\xoffset-1,\y) {$\mathbf A$};
      \node (B\l) at (\xoffset+1,\y) {$\mathbf B$};
      \node (C\l) at (\xoffset+0,\y+1) {$\mathbf C$};
      \draw (\xoffset-1-\sfpad,\y-\sfpad) rectangle (\xoffset+1+\sfpad,\y+1+\sfpad);
      \draw[dashed,-latex,thick] (0,0-\sfpad) -- (\xoffset,\y+1+\sfpad);
    }
    % graph 6
    \draw[-stealth] (A6) to[out=45+\sfdeg,in=180+45-\sfdeg] (C6);
    \draw[-stealth] (C6) to[out=180+45+\sfdeg,in=45-\sfdeg] (A6);
    \draw[-stealth] (B6) to[out=90+45+\sfdeg,in=180+90+45-\sfdeg] (C6);
    \draw[-stealth] (C6) to[out=180+90+45+\sfdeg,in=90+45-\sfdeg] (B6);
    \draw[-stealth] (A6) to[out=\sfdeg,in=180-\sfdeg] (B6);
    \draw[-stealth] (B6) to[out=180+\sfdeg,in=-\sfdeg] (A6);
    % graph 5
    \draw[-stealth] (A5) -- (C5);
    \draw[-stealth] (A5) -- (B5);
    \draw[-stealth] (B5) to[out=90+45+\sfdeg,in=180+90+45-\sfdeg] (C5);
    \draw[-stealth] (C5) to[out=180+90+45+\sfdeg,in=90+45-\sfdeg] (B5);
    % graph 4
    \draw[-stealth] (A4) -- (C4);
    \draw[-stealth] (A4) -- (B4);
    \draw[-stealth] (B4) -- (C4);
    % graph 1
    \draw[-stealth] (A1) -- (C1);
    \draw[-stealth] (B1) -- (C1);
    % graph 2
    \draw[-stealth] (B2) -- (C2);
    \draw[-stealth] (B2) -- (A2);
    \draw[-stealth] (A2) -- (C2);
    % graph 3
    \draw[-stealth] (B3) -- (A3);
    \draw[-stealth] (B3) -- (C3);
    \draw[-stealth] (A3) to[out=45+\sfdeg,in=180+45-\sfdeg] (C3);
    \draw[-stealth] (C3) to[out=180+45+\sfdeg,in=45-\sfdeg] (A3);
    % L
    %
    \def\y{-5}
    \foreach \xoffset/\l/\A/\B in {+5/12/C/B, +3/7/C/B, +1/8/C/B, -1/10/A/C, -3/9/A/C, -5/11/A/C} {
      \node (\A\l) at (\xoffset-0.5,\y) {$\mathbf \A$};
      \node (\B\l) at (\xoffset+0.5,\y) {$\mathbf \B$};
      \draw[dashed] (\xoffset-0.5-\sfpad,\y-\sfpad) rectangle (\xoffset+0.5+\sfpad,\y+\sfpad);
    }
    % graph 12
    \draw[dashed,-latex,thick] ($ (A5) + (1,-\sfpad) $) -- ($ (C12) + (0.5,\sfpad) $);
    \draw[-stealth] (C12) -- (B12);
    % graph 7
    \draw[dashed,-latex,thick] ($ (A1) + (1,-\sfpad) $) -- ($ (C7) + (0.5,\sfpad) $);
    \draw[dashed,-latex,thick] ($ (A4) + (1,-\sfpad) $) -- ($ (C7) + (0.5,\sfpad) $);
    \draw[dashed,-latex,thick] ($ (A5) + (1,-\sfpad) $) -- ($ (C7) + (0.5,\sfpad) $);
    % graph 8
    \draw[dashed,-latex,thick] ($ (A1) + (1,-\sfpad) $) -- ($ (C8) + (0.5,\sfpad) $);
    \draw[dashed,-latex,thick] ($ (A4) + (1,-\sfpad) $) -- ($ (C8) + (0.5,\sfpad) $);
    \draw[dashed,-latex,thick] ($ (A5) + (1,-\sfpad) $) -- ($ (C8) + (0.5,\sfpad) $);
    \draw[-stealth] (B8) -- (C8);
    % graph 10
    \draw[dashed,-latex,thick] ($ (A1) + (1,-\sfpad) $) -- ($ (A10) + (0.5,\sfpad) $);
    \draw[dashed,-latex,thick] ($ (A2) + (1,-\sfpad) $) -- ($ (A10) + (0.5,\sfpad) $);
    \draw[dashed,-latex,thick] ($ (A3) + (1,-\sfpad) $) -- ($ (A10) + (0.5,\sfpad) $);
    \draw[-stealth] (A10) -- (C10);
    % graph 9
    \draw[dashed,-latex,thick] ($ (A1) + (1,-\sfpad) $) -- ($ (A9) + (0.5,\sfpad) $);
    \draw[dashed,-latex,thick] ($ (A2) + (1,-\sfpad) $) -- ($ (A9) + (0.5,\sfpad) $);
    \draw[dashed,-latex,thick] ($ (A3) + (1,-\sfpad) $) -- ($ (A9) + (0.5,\sfpad) $);
    \draw[-stealth] (C11) -- (A11);
    % graph 9
    \draw[dashed,-latex,thick] ($ (A3) + (1,-\sfpad) $) -- ($ (A11) + (0.5,\sfpad) $);
    \def\y{-7}
    \foreach \xoffset/\l/\A in {+1.5/13/B,0/14/C,-1.5/15/A} {
      \node (\A\l) at (\xoffset,\y) {$\mathbf \A$};
      \draw (\xoffset-\sfpad,\y-\sfpad) rectangle (\xoffset+\sfpad,\y+\sfpad);
    }
    \draw[dashed,-latex,thick] ($ (C7) + (0.5,-\sfpad) $) -- ($ (B13) + (0,\sfpad) $);
    \draw[dashed,-latex,thick] ($ (C12) + (0.5,-\sfpad) $) -- ($ (B13) + (0,\sfpad) $);
    \draw[dashed,-latex,thick] ($ (C7) + (0.5,-\sfpad) $) -- ($ (C14) + (0,\sfpad) $);
    \draw[dashed,-latex,thick] ($ (C8) + (0.5,-\sfpad) $) -- ($ (C14) + (0,\sfpad) $);
    \draw[dashed,-latex,thick] ($ (A9) + (0.5,-\sfpad) $) -- ($ (C14) + (0,\sfpad) $);
    \draw[dashed,-latex,thick] ($ (A10) + (0.5,-\sfpad) $) -- ($ (C14) + (0,\sfpad) $);
    \draw[dashed,-latex,thick] ($ (A9) + (0.5,-\sfpad) $) -- ($ (A15) + (0,\sfpad) $);
    \draw[dashed,-latex,thick] ($ (A11) + (0.5,-\sfpad) $) -- ($ (A15) + (0,\sfpad) $);
  \end{tikzpicture}
  \caption{%
    Superflow of the example given in Fig.~\ref{subfig:n4gg}.
  }
  \label{fig:exampleNEW2}
\end{figure}
\end{document}